\documentclass[11pt,a4paper,reqno]{amsart}
\pdfoutput=1
\usepackage{amssymb}
\usepackage[USenglish]{babel}
\usepackage{palatino,eulervm}
\usepackage[T1]{fontenc}
\allowdisplaybreaks[4]

\newcommand{\zero}{\phantom{\,0\,}}

\newtheorem{thm}{Theorem}
\newtheorem{lemma}[thm]{Lemma}
\newtheorem{prop}[thm]{Proposition}
\newtheorem{rem}[thm]{Remark}

\newtheorem{df}[thm]{Definition}

\newcommand{\R}{\mathbb{R}}
\newcommand{\C}{\mathbb{C}}
\newcommand{\Q}{\mathbb{H}}
\newcommand{\arxiv}{arXiv:}
\newcommand{\tr}{\mathrm{Tr}}
\newcommand{\Cl}{\mathcal{C}\ell}
\newcommand{\ve}[2]{\genfrac{[}{]}{0pt}{}{\,#1\,}{#2}}

\newcommand{\secondorder}{2nd order condition}
\newcommand{\morita}{Hodge property}

\begin{document}

\title[Fundamental fermions as noncommutative differential forms]{\vspace*{-1.2cm}The Standard Model in noncommutative geometry:\\[5pt] fundamental fermions as internal forms}

\author[L.~D{\k a}browski]{Ludwik D{\k a}browski}
\address[L.~D{\k a}browski]{Scuola Internazionale Superiore di Studi Avanzati (SISSA), via Bonomea 265, I-34136 Trieste}
\email{dabrow@sissa.it}

\author[F.~D'Andrea]{Francesco D'Andrea} 
\address[F.~D'Andrea]{Universit\`a di Napoli ``Federico II'' and I.N.F.N. Sezione di Napoli, Complesso MSA, Via Cintia, 80126 Napoli, Italy.}
\email{francesco.dandrea@unina.it}

\author[A.~Sitarz]{Andrzej Sitarz}
\address[A.~Sitarz]{Instytut Fizyki Uniwersytetu Jagiello{\'n}skiego, Stanis{\l}awa {\L}ojasiewicza 11, 30-348 Krak{\'o}w, Poland.}
\email{sitarz@if.uj.edu.pl}



\begin{abstract}
Given the algebra, Hilbert space $H$, grading and real structure of the finite spectral triple of the Standard Model, we classify all possible Dirac operators such that $H$ is a self-Morita equivalence bimodule for the associated Clifford algebra.
\end{abstract}

\maketitle

\vspace*{-3mm}

\section{Introduction}

The noncommutative geometry approach to the Standard Model of elementary particles (reviewed for example in \cite{CM08} or in \cite{vS15}) has several interesting features.
It is purely geometric, and it widens the old idea of Kaluza-Klein theory by employing a discretized (actually finite) noncommutative internal space $F$.
This explains the lack of direct observability of such internal space, without adducing the small compactification radius that leads to a tower of unobserved fields. 
By choosing the appropriate internal noncommutative space, one gets a version of the Standard Model with neutrino mixing, with the correct particle content and coupled with (classical) gravity. 
The dynamics is ruled by the spectral action principle and also the quantized model has been worked out, cf. \cite{CM08}.

The mathematical framework is that of \emph{almost commutative} spectral triples, that means products of the canonical spectral triple on a closed Riemannian spin manifold $M$
and a finite-dimensional spectral triple.
For the Standard Model, the latter is built on the real $C^*$-algebra $\C \oplus \mathbb{H} \oplus M_3(\C)$, represented on a $96$-dimensional Hilbert space encoding the
internal degrees of freedom of the multiplet of fermions.
While the Dirac operator on $M$ is canonical, the axioms of a real spectral triple leave a lot of freedom in the choice of Dirac operator for the non-commutative internal space, and not all such operators are physically admissible. For example, fluctuations of some of these Dirac operators may give rise to bosons carrying an interaction between leptons and quarks (the so-called \emph{leptoquarks}), which are not observed experimentally and have some unpleasant features, for example they break the $SU(3)$ symmetry of the Standard Model \cite{PSS99} (in earlier models this was avoided by 
imposing commutation of the Dirac operator with an auxiliary grading, called $S^0$-real structure \cite[Page 6206]{Con95}). In most recent models \cite{CM08}, the leptoquark terms were removed simply by putting them to zero by hand.
It would be clearly desirable to find some geometric condition that reduces the arbitrariness in the choice of Dirac operator 
(see \cite{BF16, BBB15} for first steps in this direction and also \cite[\S2.6]{CCM07} for the ``massless photon'' condition).

On an oriented Riemannian manifold $(M,g)$ there are (at least) two natural choices of Dirac-type operator.	One is the Hodge-de Rham operator, acting on a space of differential forms. If $M$ is spin, the other choice is the canonical Dirac operator of the spin structure, acting on the space of ``Dirac spinors''. Both these two choices can be characterized in terms of the algebra $\Cl(M,g)$ of sections of the Clifford bundle. If $M$ is spin$^c$, spinors form a Morita equivalence $C(M)$-$\Cl(M,g)$ bimodule; if in addition $M$ is spin, then the immersion of $\Cl(M,g)$ into the commutant of $C(M)$ is realized by conjugation by an antilinear isometry $J$, the ``real structure''.

Instead, in the example of Hodge-de Rham operator the commutant of $\Cl(M,g)$ contains not only $C(M)$, but a copy of $\Cl(M,g)$ itself: in fact, the module of differential forms is a $\Cl(M,g)$ self-Morita equivalence bimodule.

The Clifford algebra $\Cl(M,g)$ has a natural generalization in the framework of spectral triples (recalled in \S\ref{sec:2}): given a spectral triple $(A,H,D)$, we denote by $\Cl_D(A)$ the $C^*$-subalgebra of $\mathcal B(H)$ generated by $A$ and by $1$-forms $[D,a]$, $a\in A$.
In \cite{DD14}, inspired by the algebraic characterization of Dirac spinors, we investigated the constraints on the finite non-commutative space of the Standard Model coming from the request that $H$ is a Morita equivalence $A$-$\Cl_D(A)$ bimodule.
This condition, however, is not satisfied by the Dirac operator proposed by Chamseddine and Connes \cite{CM08}: one needs to modify the grading and allow extra terms in the Dirac operator, which produce additional bosonic fields (carrying an interaction between neutrinos and electrons and between neutrinos and quarks). Physical consequences of such a modification are under investigation \cite{DKL17}.
For the finite non-commutative space of the Standard Model and the internal Dirac operator proposed by Chamseddine and Connes \cite{CM08} the commutant of $\Cl_D(A)$ contains not only $A$, but a copy of $\Cl_D(A)$: 
there are two commuting representations of $\Cl_D(A)$ on $H$, transformed one into the other by the charge conjugation operator $J$. This is the so-called \secondorder\ of \cite{FB14a}, where the Dirac operators satisfying it were divided into four classes.

However, it is interesting determine whether the commutant of $\Cl_D(A)$ is a copy of $\Cl_D(A)$ itself, like in the Hodge-de Rham case, or is bigger.
In the former case we shall say that the \morita\ holds.

In this paper we investigate this question
and shed further light on the underlying geometric and algebraic structure of the noncommutative geometry approach to the Standard Model of elementary particles.
Therein the number of generations of particles,
three as physically affirmed, is in principle arbitrary, and we focus on the simplest case of one generation.

We prove that, under some non-degeneracy conditions (that are generically satisfied by Dirac operators in the aforementioned four classes), $H$ is indeed a self-Morita equivalence $\Cl_D(A)$-bimodule, exactly like in the case of Hodge-de Rham operator on a manifold. Thus the internal degrees of freedom of elementary fermions are in this sense described by differential forms on a finite non-commutative space. This provides also a geometrical meaning to the \secondorder\
as a pre-\morita.

The paper is organized as follows. In \S\ref{sec:2} we recall some preliminary notions about spectral triples, introduce notations and explain the property we are interested in, that we name \morita.
In \S\ref{sec:3} we collect some general results that are valid for a finite-dimensional spectral triple. In \S\ref{sec:4} we introduce the finite noncommutative space of the Standard Model and, using some general lemmas of \S\ref{sec:3}, reproduce a result of \cite{FB14a} on the \secondorder\ (cf.~Prop.~\ref{prop:5}). In \S\ref{sec:5} we prove our main result, i.e.~that in the Standard Model case almost all Dirac operators satisfy the \morita.
Finally in \S\ref{sec:6} we comment on the case of more than one generation of particles.

We adopt the following notations: if $a\in\C$, we denote by $\overline{a}$ the complex conjugated; if $a\in M_n(\C)$ is a matrix, we denote by $a^*$ the Hermitian conjugated and by $\overline{a}=(a^*)^t$ the entry-wise complex conjugated.

\section{Preliminaries}\label{sec:2}

\begin{df}
A unital \emph{spectral triple} $(A,H,D)$ is the datum of:
\begin{itemize}\itemsep=0pt
\item[(i)] a (real or complex) unital $*$-algebra $A$ of bounded operators on a (separable) complex Hilbert space $H$,
\item[(ii)] a selfadjoint operator $D$ on $H$ with compact resolvent,
\end{itemize}
such that $a\cdot\mathrm{Dom}(D)\subseteq\mathrm{Dom}(D)$ and $[D,a]$ extends to a bounded operator on $H$, for all $a\in A$.
The spectral triple is:
\begin{description}
\item[even] if there exists a grading operator $\gamma$ on $H$, i.e.~$\gamma=\gamma^*$ and $\gamma^2=1$, that commutes with $A$ and anticommutes with $D$;
\item[real] if there exists an antilinear isometry $J$ on $H$ s.t.
\begin{equation}\label{eq:KOdim}
J^2=\varepsilon 1
\;,\qquad
JD=\varepsilon' DJ
\quad
\text{and (only in the even case)}\quad
J\gamma=\varepsilon'' \gamma J
\;,
\end{equation}
for some $\varepsilon,\varepsilon',\varepsilon''=\pm 1$, and $\forall\;a,b\in A:$
$$
\begin{array}{c}
[a,JbJ^{-1}]=0, \\
\text{{\small (reality)}}
\end{array}
\qquad\quad
\begin{array}{c}
[[D,a],JbJ^{-1}]=0 .\\

\text{{\small (1st order)}}
\end{array}
$$
\end{description}
\end{df}

\noindent
The three signs in
\eqref{eq:KOdim} define the KO-dimension of the spectral triple (see e.g.~\cite{CM08}).

A typical example of a spectral triple arises from differential geometry. Let $(M,g)$ be an oriented closed Riemannian manifold, $E\to M$ a Hermitian vector bundle equipped with a unitary Clifford action $c:\Gamma^\infty(M,T^*_{\C}M\otimes E)\to \Gamma^\infty(M,E)$ on smooth sections and a connection $\nabla^E$ compatible with $g$. Then:
$$
A=C^\infty(M) \qquad\quad H=L^2(M,E) \qquad\quad D=c\circ\nabla^E
$$
is a spectral triple. 
Two main examples belonging to this class are:
\begin{list}{$\bullet$}{\itemsep=2pt \leftmargin=2em}
\item the Hodge operator $D=d+d^*$ on $E=\bigwedge^\bullet T^*_{\C}M$;
\item the Dirac operator $D=D\hspace{-8pt}\boldsymbol/\,$ on the spinor bundle $E$ (if $M$ is a spin manifold).
\end{list}

The meaning of the first order condition is that $D$ is a first order differential operator acting on smooth sections of a vector bundle.

Let $\Cl(M,g)$ be the algebra of (continuous) sections of the Clifford bundle of $M$:
as a $C(M)$-module, this is isomorphic to the module of continuous sections of the bundle $\bigwedge^\bullet\hspace{1pt}T^*_{\C}M\to M$, but with product defined by the Clifford multiplication.
In the above class of examples, $H$ carries \emph{commuting} representations of the algebras $C(M)$ and $\Cl(M,g)$.
In fact, in the spin manifold example, $\Gamma(M,E)$ is a Morita equivalence $\Cl(M,g)$-$C(M)$ bimodule, and it is well known that a closed oriented Riemannian manifold $M$ admits a spin$^c$ structure if and only if a Morita equivalence $\Cl(M,g)$-$C(M)$ bimodule exists (see e.g.~\S1 of \cite{Var06}, or the original paper \cite[\S2]{Ply86}).
In the Hodge example, on the other hand, $\Gamma(M,E)$ is a $\Cl(M,g)$ self-Morita equivalence bimodule, that is: the algebra of $\Cl(M,g)$-linear adjointable endomorphisms of $\Gamma(M,E)$ is isomorphic to $\Cl(M,g)$ itself.

\smallskip

In this paper we are interested in finite-dimensional
spectral triples, that is: we assume that $H$ is finite-dimensional. As a consequence, $\mathcal{B}(H)=\mathrm{End}_{\C}(H)$, $A\subseteq\mathrm{End}_{\C}(H)$ is finite-dimensional as well,
and the compact resolvent condition for $D$ is automatically satisfied.

\begin{df}\label{def:2}
Let $(A,H,D,J)$ be a finite-dimensional real spectral triple.
We set:
\begin{list}{}{\leftmargin=2em \itemsep=3pt}
\item
$\xi^\circ:=J\xi^* J^{-1}, \quad \forall\;\xi \in \mathrm{End}_{\C}(H)$,
\end{list}
If $B$ is a subset of $\mathrm{End}_{\C}(H)$, we call:
\begin{list}{}{\leftmargin=2em \itemsep=3pt}
\item
$B^\circ :=\{\xi^\circ: \xi\in B \}$;

\item
$B' :=\{\xi\in\mathrm{End}_{\C}(H):[\xi,\eta]=0\;\forall\;\eta\in B\}$ 

(the \emph{commutant} of $B$).
\end{list}
Finally, we define:
\begin{list}{}{\leftmargin=2em \itemsep=3pt}
\item
 $\Omega^1_D(A)$ as the \emph{complex} vector subspace of $\mathrm{End}_{\C}(H)$ spanned by $a[D,b]$, $a,b\in A$;
\item
$\Cl_D(A)$ as the complex $*$-subalgebra of 
$\mathrm{End}_{\C}(H)$ generated by $A$ and $\Omega^1_D(A)$.
\end{list}
\end{df}

As usual, we refer to $\Omega^1_D(A)$ as the $A$-bimodule of differential $1$-forms, and we call $\Cl_D(A)$ the \emph{Clifford algebra} of the spectral triple. 
Throughout this paper we assume that 
$\Omega^1_D(A)\neq \{0\}$.

Motivated by the above discussion, we are interested in the following two conditions:

\begin{df}\label{def:3}
We say that a real spectral triple $(A,H,D,J)$ satisfies the \emph{\secondorder} if
\begin{itemize}
\item[(i)] \ $\Cl_D(A)$ and $\Cl_D(A)^\circ$ commute, that is: $\Cl_D(A)^\circ\subseteq\Cl_D(A)'$.
\end{itemize}
We say that the \emph{\morita} holds if
\begin{itemize}
\item[(ii)] \ $\Cl_D(A)'=\Cl_D(A)^\circ$.
\end{itemize}
\end{df}

Obviously (ii) implies (i). Such a condition is the statement that elements of $H$ are the analogue of differential forms (and $\Cl_D(A)$ and $\Cl_D(A)^\circ$ act on $H$ by ``left'' and ``right'' Clifford multiplication), in contrast with the Morita condition in \cite{DD14}, which would say that the elements of $H$ are the analogue of Dirac spinors.

\begin{rem}
Condition (i) for a real spectral triple is equivalent to
\begin{equation}\label{eq:2ndorder}
[[D,a],[D,b]^\circ]=0
\qquad\forall\;a,b\in A.
\end{equation}
Indeed, clearly (i) implies $\Omega^1_D(A)^\circ\subseteq\Omega^1_D(A)'$, that is \eqref{eq:2ndorder}. But $A$ and $\Omega^1_D(A)$ generate the Clifford algebra, and $A^\circ\subseteq\Cl_D(A)'$
by the reality and 1st order condition. So, \eqref{eq:2ndorder} implies (i) as well.
\end{rem}

The terminology \emph{\secondorder} was introduced in \cite{FB14a} and refers to the vanishing of
\eqref{eq:2ndorder}. Note however that the terminology is a bit misleading since, while the 1st order condition says that $D$ is the analogue of a 1st order differential operator, the \secondorder\ doesn't refer, of course, to $D$ being of 2nd order.

\begin{rem}
In the case of an even spectral triple, an alternative definition of $\Cl_D(A)$ could be as the algebra generated by $A$, $\Omega^1_D(A)$ and $\gamma$. Note however that, with such a definition, condition \textup{(i)} becomes too strong and is satisfied only in trivial cases: namely, since $\gamma= \pm \gamma^\circ$, the grading would both commute (condition (i)) and anticommute (by definition of even spectral triple) with elements of $\Omega^1_D(A)$, which implies $\Omega^1_D(A)\equiv 0$.
\end{rem}

\section{Finite-dimensional spectral triples}\label{sec:3}
Let $(A,H,D)$ be a finite-dimensional spectral triple, $J$ an antilinear isometry satisfying \eqref{eq:KOdim} and the reality condition, and call $A_{\C}$ the complex $*$-subalgebra of $\mathrm{End}_{\C}(H)$ generated by $A$ (so $A=A_{\C}$ if $A$ is already complex).

From the structure theorem for finite-dimensional $C^*$-algebras:
$$
A_{\C}\simeq\bigoplus_{i=1}^NM_{n_i}(\C) \;.
$$
If $P_i$ is the unit of the summand $M_{n_i}(\C)$,
then $P_1,\ldots,P_N$ are (represented by) orthogonal projections on $H$ whose sum is $1$.
From the reality condition it follows that the operators $Q_i := J P_i J^{-1}$ form a set of orthogonal projections as well, commuting with the projections $P_i$'s, and whose sum is also $1$.

Calling
$$
D_{ij, kl}:=P_iQ_jDP_kQ_l \;,
$$
we can decompose the Dirac operator into four pieces:\footnote{For a general spectral triple, despite the notations, this is not the same decomposition that appears in \cite[\S4]{DD14}, although in the Standard Model case the two decompositions coincide.}
\begin{equation}\label{eq:D0D1DR}
D=D_0+D_1+D_2+D_R \;,
\end{equation}
where
\begin{align*}
D_0 &:= \sum_{i,j,k\,:\,i \neq k} D_{ij,kj} \;, &
D_1 &:= \sum_{i,j,l\,:\,j \neq l} D_{ij, il} \;, \\[2pt]
D_2 &:= \sum_{\substack{i,j,l,k \\[1pt] i\neq k,j \neq l}} D_{ij, kl} \;,&
D_R &:= \sum_{i,j} D_{ij, ij} \;.
\end{align*}
Note that \eqref{eq:KOdim} implies
$JD_{ij, kl}J^{-1}=\varepsilon'\,D_{ji,lk}$ (for all $i,j,k,l$) and then
\begin{equation}\label{eq:JDiJ}
JD_0J^{-1}= \varepsilon'\,D_1 \;,\qquad
JD_2J^{-1} = \varepsilon'\,D_2 \;,\qquad
JD_RJ^{-1} = \varepsilon'\,D_R \;.
\end{equation}
Note also that \mbox{$(D_{ij, kl})^*=D_{kl,ij}$}, which means that all the summands in \eqref{eq:D0D1DR} are selfadjoint operators.

Let us extend some of the results in \cite{PS96} and  add some considerations about the \secondorder\ and \morita.

\begin{lemma}\label{lemma:6}
$D_0+D_2\in\Omega^1_D(A)$.
\end{lemma}
\begin{proof}
An explicit computation gives
$D_0+D_2=\sum_{i\neq k}P_iDP_k=\sum_{i\neq k}P_i[D,P_k]$, where we used the fact that $P_iP_k=0$ for $i\neq k$.
\end{proof}

\begin{prop}[1st order]\label{prop:6}
$D$ satisfies the 1st order condition if and only if: $\,D_2=0$,
$D_1\in A'$, and $D_R$ satisfies the 1st order condition.
\end{prop}

\begin{proof}
For $i\neq k$ and $j\neq l$, since $P_iP_k=Q_jQ_l=0$, if $D$ satisfies the 1st order condition, one has:
$$
D_{ij,kl}=P_i[Q_j,[D,P_k]]Q_l=0 \;,
$$
that means $D_2=0$. Similarly for $i=k$ and $j\neq l$,
since $P_i$ is central in $A_{\C}$ we get
$$
[D_{ij,il},a]=P_i[Q_j,[D,a]]P_iQ_l=0 \qquad\forall\;a\in A,
$$
that means $D_1\in A'$.

Note that if $D_2=0$ and $D_1\in A'$, for all $a\in A$ and $b^\circ\in A^\circ$ one has
$$
[[D,a],JbJ^{-1}]=[[D_0,a],JbJ^{-1}]+[[D_R,a],JbJ^{-1}]
$$
But
$$
[[D_0,a],JbJ^{-1}]=\epsilon' J[[D_1,JaJ^{-1}],b]J^{-1}
$$
is zero since both $D_1$ and $JaJ^{-1}$ commute with $b$ (cf.~reality condition). Hence $[[D,a],JbJ^{-1}]=[[D_R,a],JbJ^{-1}]$ and $D$ satisfies the 1st order condition if and only if $D_R$ does.
\end{proof}

\noindent
Note that:
\begin{itemize}\itemsep=2pt
\item $A$ and $A_{\C}$ have the same commutant in $\mathrm{End}_{\C}(H)$;
\item due to \eqref{eq:JDiJ}, $D_1\in A'\Leftrightarrow D_0\in (A^\circ)'$ and $D_R\in A' \Leftrightarrow D_R\in (A^\circ)'$.
\end{itemize}

\begin{rem}\label{rem:8}
If $(A,H,D,J)$ is a finite-dimensional real spectral triple, from Prop.~\ref{prop:6} and Lemma~\ref{lemma:6} we deduce that $D_0\in\Omega_D^1(A)\subseteq\Cl_D(A)$. If in addition $D_R\in A'$, then $[D,a]=[D_0,a]\;\forall\;a\in A$ and the Clifford algebra is generated by $A_{\C}$ and $D_0$.
\end{rem}

\begin{prop}[2nd order]\label{prop:9}
Let $(A,H,D,J)$ be a finite-dimensional real spectral triple with $D_R\in A'$. The \secondorder\ is satisfied if and only if $[D_0,D_1]=0$.
\end{prop}

\begin{proof}
From Rem.~\ref{rem:8}, $\Cl_D(A)$ is generated by $A_{\C}$ and $D_0$, and $\Cl_D(A)^\circ$ by $A_{\C}^\circ$ and $D_1$. Now, $D_0$ commutes with $A_{\C}^\circ$, $D_1$ with $A_{\C}$ (Prop.~\ref{prop:6}) and $A_{\C}$ with $A_{\C}^\circ$ (reality condition).
The algebras $\Cl_D(A)$ and $\Cl_D(A)^\circ$ are then mutually commuting if and only if $D_0$ commutes with $D_1$.
\end{proof}

\begin{rem}
It is easy to produce examples in which $D_R$ does not satisfy the 1st order condition, or examples of spectral triples satisfying the 1st order condition and the \secondorder, but with $D_R\notin A'$ (proving that such a condition is sufficient but not necessary).
Take $A=H=M_n(\C)$, with representation given by left multiplication, and \mbox{$J(a)=a^*$} the Hermitian conjugation. Then:
\begingroup
\setlength{\leftmargini}{1.9em} 
\begin{itemize}\itemsep=2pt
\item[(a)] if $D=dd^\circ$ with $d=d^*\in A$, in the notations above one has $D_R=D_{11,11}=D$, and the 1st order condition is not satisfied unless $d$ is central;

\item[(b)] if $D=d+d^\circ$ with $d=d^*\in A$, the 1st and 2nd order conditions are satisfied, again $D_R=D_{11,11}=D$, and $D_R\notin A'$ unless $d$ is central.
\end{itemize}
\endgroup
\end{rem}

\begin{rem}\label{rem:11}
Note that $D_R$ maps each vector subspace $H_{ij}=P_iQ_jH$ to itself, as well as $\gamma$ in the even case. If on a subspace $H_{ij}$ the grading $\gamma$ is proportional to the identity, this forces $D_R$ to be zero on $H_{ij}$.

If the spectral triple is even and orientable \cite[Axiom 4']{{Con96}} then $\gamma$ is proportional to the identity on each subspace $H_{ij}$ (see \cite[Lemma 3]{PS96} or \cite[\S3.3]{Kra97}) and $D_R=0$.
\end{rem}

We pass now to the \morita. Let us start with some considerations about the commutant of a finite-dimensional $C^*$-algebra.

\medskip

Let $B$ be a unital complex $*$-subalgebra of $\mathrm{End}_{\C}(H)$. From the structure theorem for finite-dimensional $C^*$-algebras we know that it is a finite direct sum of matrix algebras:
$B\simeq\bigoplus_{i=1}^sM_{m_i}(\C)$ for some positive integers $s,m_1,\ldots,m_s$. Call~$\widetilde{P}_i$ the unit of the $i$-th summand $M_{m_i}(\C)$,
then $\widetilde{P}_1,\ldots,\widetilde{P}_s$ are orthogonal projections and $H$ decomposes as
$H\simeq\bigoplus_{i=1}^sH_i$, with
\begin{equation}\label{eq:Hi}
\widetilde{H}_i = \widetilde{P}_i\cdot H \simeq \C^{m_i} \otimes \C^{k_i},  \;,
\end{equation}
where $k_i$ is the multiplicity of the (unique) irreducible representation $\C^{m_i}$ of $M_{m_i}(\C)$ 
in $\widetilde{P}_i\cdot H$, and the action of the algebra $M_{m_i}(\C)$ on $ \C^{m_i} \otimes \C^{k_i}$ 
is given by matrix multiplication on the first component of the tensor product.
With this notations, we can now compute the commutant $B'$  of $B$.

\begin{lemma}\label{lemma:11}
$B'\simeq\bigoplus\nolimits_{i=1}^sM_{k_i}(\C)$ and the action of $B'$ on 
$\widetilde{H}_i \simeq \C^{m_i} \otimes \C^{k_i}$ is given by matrix multiplication on the second component of the tensor product by $M_{k_i}(\C)$.
\end{lemma}

\begin{proof}
Let $f\in\mathrm{Hom}_B(H, H)$ be a $B$-linear map. Then $\widetilde{P}_j f(v)=f(\widetilde{P}_j v)=0$ 
for all $v\in H$ and all $j$ so  $f(\widetilde{H}_j)\subseteq\widetilde{H}_j$ and therefore
$$
\mathrm{End}_B(H)=\bigoplus\nolimits_{i=1}^s\mathrm{Hom}_B(\widetilde{H}_i,\widetilde{H}_i) 
= \bigoplus\nolimits_{i=1}^s\mathrm{End}_B(\widetilde{H}_i) \;.
$$
The rest of the proof is straightforward: the commutant of $M_{m_i}(\C)$ in $\C^{m_i} \otimes \C^{k_i}$ is given by $M_{k_i}(\C)$ acting on the second component of the tensor product, and $B'=\mathrm{End}_B(H)\simeq\bigoplus\nolimits_{i=1}^sM_{k_i}(\C)$.
\end{proof}

To check the \morita\ in the Standard Model case we will use the following simple observation.

\begin{lemma}\label{lemma:12}
Let $(A,H,D,J)$ be a finite-dimensional real spectral triple and 
$B\subseteq\mathrm{End}_{\C}(H)$ a unital complex $*$-algebra satisfying:
$$
B'=B^\circ \qquad\text{and}\qquad
\Cl_D(A)\subseteq B \;.
$$
The following are equivalent:\vspace{5pt}
\begin{center}
\textup{(a)}\, the \morita\ is satisfied; \qquad
\textup{(b)}\, $\Cl_D(A)'\subseteq B^\circ$; \qquad
\textup{(c)}\, $\Cl_D(A)=B$.
\end{center}
\end{lemma}

\begin{proof}~\vspace{-2pt}
\begin{list}{}{\leftmargin=1em \itemsep=3pt}
\item[(a) $\Rightarrow$ (b)]
$\Cl_D(A)\subseteq B$ implies $\Cl_D(A)^\circ \subseteq B^\circ$. If \morita\	holds then $\Cl_D(A)' =  \Cl_D(A)^\circ \subseteq B^\circ$.

\item[(b) $\Rightarrow$ (c)]
The condition $\Cl_D(A)'\subseteq B^\circ = B'$ implies $B \subseteq\Cl_D(A)$ and, since the opposite inclusion holds by hypothesis, $\Cl_D(A)=B$.

\item[(c) $\Rightarrow$ (a)]
If (c) holds,	then $B'=B^\circ$ translates into $\Cl_D(A)'=\Cl_D(A)^\circ$.\vspace{-19pt}
\end{list}
\end{proof}

The internal Dirac operators for the Standard Model satisfying the \secondorder\ can be divided in four (not disjoint) classes, cf.~Prop.~\ref{prop:5}.
The strategy to check the \morita\ in the Standard Model case, in each of the four aforementioned cases, will be the following:\label{strategy}
\begingroup
\setlength{\leftmargini}{1.8em} 
\begin{enumerate}\itemsep=2pt
\item
We define a suitable big algebra $B$ containing the Clifford algebra $\Cl_D(A)$, which is
independent of $D$.
\item
We check that $B$ and $B^\circ$ commute, and so $B^\circ\subseteq B'$.
\item
We prove that $B^\circ = B'$ using the dimensional arguments: since, the algebras are finite-dimensional (as vector spaces) one has just to check that $B^\circ$ and $B'$ have the same dimension. Since $\dim(B^\circ)=\dim(B)$, it suffices to check that $\dim(B)=\dim(B')$, using Lemma \ref{lemma:11}.
\item
With Lemma \ref{lemma:12} we find under what conditions on $D$ one has 
$\Cl_D(A)'\subseteq B^\circ$ or (equivalently) $\Cl_D(A) = B$. 
\end{enumerate}
\endgroup

\section{The Standard Model spectral triple}\label{sec:4}
Let us recall the data $(A,H,\gamma,J)$ describing the finite noncommutative space of the Standard Model of elementary particles.
The Dirac operator $D$ is arbitrary for the time being.
To simplify the discussion, we will consider a model with only one generation of particles. In \S\ref{sec:6} we'll comment on what changes with three (or in general more than one) generations.

We adopt the notations of \cite{DD14}.
The Hilbert subspace representing particles is $F=M_4(\C)$ with inner product $\langle a,b\rangle=\tr(a^*b)$, where $a^*$ is the Hermitian conjugated of $a$. Let $e_{ij}\in M_N(\C)$ be the matrix with $1$ in position $(i,j)$ and zero everywhere else (we will omit the size of the matrix, which should be clear from the context).
We arrange particles in a $4\times 4$ matrix in the following way:
\begin{equation}
\begin{bmatrix}
\nu_R & u^1_R & u^2_R & u^3_R \\
e_R   & d^1_R & d^2_R & d^3_R \\
\nu_L & u^1_L & u^2_L & u^3_L \\
e_L   & d^1_L & d^2_L & d^3_L
\end{bmatrix} \;.
\label{matrixrep}
\end{equation}
So, for example $e_{21}$ represents a right-handed electron, while $\frac{1}{\sqrt{2}}e_{11}+\frac{1}{\sqrt{2}}e_{41}$ is a mix of a right-handed neutrino and a left-handed electron.

The Hilbert space of our spectral triple is $H\simeq F\oplus F^*$. We write its elements in the form:
$$
H=\bigg\{\;
\ve{v}{w} \;\bigg|\, v,w\in M_4(\C) \bigg\}
$$
The real structure is given by
$$
J\ve{v}{w}=\ve{w^*}{v^*}
$$
The grading $\gamma$ on $F$ is the operator of left multiplication by the diagonal matrix
$$
\mathrm{diag}(+1,+1,-1,-1)
$$
(the ``chirality'' operator), and its action on $F^*$ is determined by the condition $\gamma J+J\gamma=0$ (the signs in \eqref{eq:KOdim} are those corresponding to KO-dimension $6$).

We identify $\mathrm{End}_{\C}(H)$ with the algebra $M_4(\C)\otimes M_2(\C)\otimes M_4(\C)$, represented on $H$ as follows:
$$
\pi(\alpha\otimes 1\otimes \beta)\ve{v}{w}=
\ve{\alpha v\beta^t}{\alpha w\beta^t}
\qquad\quad
\pi\bigg(1\otimes\text{{\small%
\setlength{\arraycolsep}{4pt}%
$
\bigg[\begin{array}{cc}
a & b \\[-2pt] c & d
\end{array}\bigg]
$}}\otimes 1\bigg)\ve{v}{w}=
\ve{ av+bw }{ cv+dw }
$$
for all $\alpha,\beta,v,w\in M_4(\C)$, $a,b,c,d\in\C$. This gives a unital $*$-representation.
With this notations:
\begin{equation}\label{eq:circle}
\xi=\pi\bigg(\alpha\otimes\text{{\small%
\setlength{\arraycolsep}{4pt}%
$
\bigg[\begin{array}{cc}
a & b \\[-2pt] c & d
\end{array}\bigg]
$}}\otimes\beta\bigg)
\mapsto\xi^\circ:=J\xi^*J^{-1}=
\pi\bigg(\beta^t\otimes\text{{\small%
\setlength{\arraycolsep}{4pt}%
$
\bigg[\begin{array}{cc}
d & b \\[-2pt] c & a
\end{array}\bigg]
$}}\otimes\alpha^t\bigg)
\end{equation}
The representation symbol $\pi$ from now on will be omitted.

The algebra $A\simeq\C\oplus\mathbb{H}\oplus M_3(\C)$ has elements
\begin{equation}\label{eq:A}
\left[\!
\begin{array}{c|c}
\begin{matrix} \;\lambda\; & \;0\;  \\ 0 & \overline\lambda \end{matrix} &
\begin{matrix} \;0\; & \;0\;  \\ 0 & 0 \end{matrix} \\
\hline
\begin{matrix} \;0\; & \;0\; \\ 0 & 0 \end{matrix} & q
\end{array}
\!\right]
\otimes e_{11}\otimes 1
+
\left[\!
\begin{array}{c|ccc}
\lambda & \;0\; & \;0\; & \;0\; \\
\hline
\begin{matrix} \;0\; \\ 0 \\ 0 \end{matrix} && m
\end{array}
\!\right]
\otimes e_{22}\otimes 1
\end{equation}
with $\lambda\in\C$, $q\in\Q$ a quaternion and $m\in M_3(\C)$.

The most general Dirac operator giving a real spectral triple of KO-dimension~$6$ can be found e.g.~in \cite[\S5.1]{DD14}. In the notations of \eqref{eq:D0D1DR} one has $D_2=0$,
$$
D_R=e_{11}\otimes(\Upsilon_Re_{21}+\overline{\Upsilon}_Re_{12})\otimes e_{11} \;,
$$
for some $\Upsilon_R\in\C$, $D_1=JD_0J^{-1}$ and:
\begin{align*}
\hspace*{-3mm}D_0 &=
\left[\!\begin{array}{cc|cc}
\zero & \zero & \alpha_{13} & \alpha_{14} \\
\zero & \zero & \alpha_{23} & \alpha_{24} \\
\hline
\overline\alpha_{13} & \overline\alpha_{23} & \zero & \zero \\
\overline\alpha_{14} & \overline\alpha_{24} & \zero & \zero
\end{array}\!\right]
\otimes e_{11}\otimes e_{11}+
\left[\!\begin{array}{cc|cc}
\zero & \zero & \beta_{13} & \beta_{14} \\
\zero & \zero & \beta_{23} & \beta_{24} \\
\hline
\overline\beta_{13} & \overline\beta_{23} & \zero & \zero \\
\overline\beta_{14} & \overline\beta_{24} & \zero & \zero
\end{array}\!\right]
\otimes e_{11}\otimes (1-e_{11}) \\
&\; +
\left[\!\begin{array}{cc|cc}
\zero & \delta_{12} & \delta_{13} & \delta_{14} \\
\delta_{21} & \delta_{22} & \delta_{23} & \delta_{24} \\
\hline
\zero & \zero & \zero & \zero \\
\zero & \zero & \zero & \zero
\end{array}\!\right]
\otimes e_{12}\otimes e_{11}+
\left[\!\begin{array}{cc|cc}
\zero & \overline\delta_{21} & \zero & \zero \\
\overline\delta_{12} & \overline\delta_{22} & \zero & \zero \\
\hline
\overline\delta_{13} & \overline\delta_{23} & \zero & \zero \\
\overline\delta_{14} & \overline\delta_{24} & \zero & \zero
\end{array}\!\right]
\otimes e_{21}\otimes e_{11}
\end{align*}
where $\alpha_{ij},\beta_{ij},\delta_{ij}\in\C$ and zeroes are omitted. Note that $D_R$, as expected from Remark \ref{rem:11}, $D_R$ is zero on all subspaces $H_{ij}$ where $\gamma$ is proportional to the identity, that is all but one: the one spanned by $\nu_R$ and $J(\nu_R)$.

Observe that $D_R\in A'$, so that the Clifford algebra is generated by $A_{\C}$ and $D_0$ (Remark \ref{rem:8}).
For future reference, let us also define:
\begin{equation}\label{eq:permmatrix}
U:=1\otimes 1\otimes 1+e_{11}\otimes (e_{12}+e_{21}-1)\otimes e_{11} \;.
\end{equation}
This is a permutation matrix: $U=U^*$ and $U^2=1$ (hence a unitary). It's action on the Hilbert space is to exchange the basis vectors $\nu_R$ and $J(\nu_R)$.

\begin{lemma}\label{lemma:14}
$U$ commutes with $A$ and $J$.
\end{lemma}

\begin{proof}
Explicit computation, in particular $JU J^{-1}=U$ is an immediate consequence of \eqref{eq:circle}.
\end{proof}

Let us now reproduce the result of \cite{FB14a}, that is Prop.~\ref{prop:5} below.

\begin{prop}\label{prop:5}
The \secondorder\ is satisfied if and only if one (at least) of the following four conditions holds:
\begin{itemize}
\item[1.] \ $\delta_{ij}=0$ for all $i,j$;
\item[2.] \ $\alpha_{13}=\alpha_{14}=0$ and $\delta_{ij}=0\;\forall\;i=1,2$ and $j=2,3,4$;
\item[3.] \ $\delta_{21} = 0$ and $\beta_{13}=\beta_{14}=0$;
\item[4.] \ $\delta_{12}=\delta_{13}=\delta_{14} = 0$, \ $\beta_{13}=\beta_{14}=0$, \ $\alpha_{13}=\alpha_{14}=0$.
\end{itemize}
\end{prop}
\begin{proof}
Let $\Delta:=D_0D_1$. From Prop.~\ref{prop:9} we know that \secondorder\ is satisfied if and only if $[D_0,D_1]=\Delta-\Delta^*$ is zero, which means $\Delta=\Delta^*$. From \eqref{eq:circle} we get
\begin{align*}
\hspace*{-2mm}D_1 &=
e_{11}\otimes e_{22}\otimes
\left[\!\begin{array}{cc|cc}
\zero & \zero & \overline\alpha_{13} & \overline\alpha_{14} \\
\zero & \zero & \overline\alpha_{23} & \overline\alpha_{24} \\
\hline
\alpha_{13} & \alpha_{23} & \zero & \zero \\
\alpha_{14} & \alpha_{24} & \zero & \zero
\end{array}\!\right]+
(1-e_{11})\otimes e_{22}\otimes
\left[\!\begin{array}{cc|cc}
\zero & \zero & \overline\beta_{13} & \overline\beta_{14} \\
\zero & \zero & \overline\beta_{23} & \overline\beta_{24} \\
\hline
\beta_{13} & \beta_{23} & \zero & \zero \\
\beta_{14} & \beta_{24} & \zero & \zero
\end{array}\!\right] \\
&\; +
e_{11}\otimes e_{12}\otimes
\left[\!\begin{array}{cc|cc}
\zero & \delta_{21} & \zero & \zero \\
\delta_{12} & \delta_{22} & \zero & \zero \\
\hline
\delta_{13} & \delta_{23} & \zero & \zero \\
\delta_{14} & \delta_{24} & \zero & \zero
\end{array}\!\right]
+e_{11}\otimes e_{21}\otimes
\left[\!\begin{array}{cc|cc}
\zero & \overline\delta_{12} & \overline\delta_{13} & \overline\delta_{14} \\
\overline\delta_{21} & \overline\delta_{22} & \overline\delta_{23} & \overline\delta_{24} \\
\hline
\zero & \zero & \zero & \zero \\
\zero & \zero & \zero & \zero
\end{array}\!\right]
\end{align*}
Let $\Delta_{ij}$ be the summand in $\Delta$ that has $e_{ij}$ as second factor.
Then $\Delta=\Delta^*$ if{}f
$\Delta_{11}$ is selfadjoint (which implies that $\Delta_{22}=(\Delta_{11})^\circ$ is also selfadjoint) and $(\Delta_{12})^*=\Delta_{21}$.

With a simple computation one finds:
$$
\Delta_{11} =
\left[\!\begin{array}{cc|cc}
\zero & \zero & \zero & \zero \\
\delta_{21} & \zero & \zero & \zero \\
\hline
\zero & \zero & \zero & \zero \\
\zero & \zero & \zero & \zero
\end{array}\!\right]
\otimes e_{11}\otimes
\left[\!\begin{array}{cc|cc}
\zero & \overline\delta_{12} & \overline\delta_{13} & \overline\delta_{14} \\
\zero & \zero & \zero & \zero \\
\hline
\zero & \zero & \zero & \zero \\
\zero & \zero & \zero & \zero
\end{array}\!\right] \;,
$$
that is selfadjoint if{}f it is zero, that is
$\delta_{21}=0$ or $\delta_{1i}=0\;\forall\;i=2,3,4$.
Moreover $\Delta_{21}=0$, while $\Delta_{12}$ is the sum of four linearly independent terms given by
\begin{subequations}\label{eq:8}
\begin{gather}
\left[\!\begin{array}{cc|cc}
\zero & \zero & \zero & \zero \\
\delta_{21} & \zero & \zero & \zero \\
\hline
\zero & \zero & \zero & \zero \\
\zero & \zero & \zero & \zero
\end{array}\!\right]
\otimes e_{12}\otimes
\left[\!\begin{array}{cc|cc}
\zero & \zero & \overline\alpha_{13} & \overline\alpha_{14} \\
\zero & \zero & \zero & \zero \\
\hline
\zero & \zero & \zero & \zero \\
\zero & \zero & \zero & \zero
\end{array}\!\right] \;,
\\
\left[\!\begin{array}{cc|cc}
\zero & \delta_{12} & \delta_{13} & \delta_{14} \\
\zero & \delta_{22} & \delta_{23} & \delta_{24} \\
\hline
\zero & \zero & \zero & \zero \\
\zero & \zero & \zero & \zero
\end{array}\!\right]
\otimes e_{12}\otimes
\left[\!\begin{array}{cc|cc}
\zero & \zero & \overline\beta_{13} & \overline\beta_{14} \\
\zero & \zero & \zero & \zero \\
\hline
\zero & \zero & \zero & \zero \\
\zero & \zero & \zero & \zero
\end{array}\!\right] \;,
\end{gather}
\end{subequations}
and the image of these two operators through \eqref{eq:circle}.
So $\Delta_{12}=(\Delta_{21})^*=0$ if and only if the elements in \eqref{eq:8} vanish, that gives us the four conditions in Prop.~\ref{prop:5}.
\end{proof}

For future reference we recall that \cite[Lemma 7]{DD14}:
\begin{equation}\label{eq:setbelow}
\begin{split}
A'=\mathrm{Span}\big\{\,
e_{22}\otimes e_{11}
\,,\,
(e_{33}+e_{44})\otimes e_{11}
\,,\,
(1-e_{11})\otimes e_{22}
\big\}\otimes M_4(\C)\;\;\oplus
\\
\oplus\;\; \C e_{11}\otimes M_2(\C)\otimes M_4(\C)
 \;.
\end{split}
\end{equation}

\section{The \morita}\label{sec:5}
In this section we characterize, among the Dirac operators satisfying the \secondorder, those satisfying the \morita\ as well.
We will see that, in the vector space parametrizing Dirac operators satisfying the \secondorder, those not satisfying the \morita\ form a measure zero subset.

As explained in \S\ref{sec:3} (page \pageref{strategy}), we start by defining some ``big'' algebra $B$ satisfying the conditions in Lemma \ref{lemma:12}. In fact, we need two of them, to cover the four cases in Prop.~\ref{prop:5}.

\subsection{Cases 1 and 2.} Let
\begin{equation}\label{eq:Bprime}
B:=\C\oplus M_3(\C)\oplus M_4(\C)\oplus M_4(\C)
\end{equation}
and $\pi$ the representation on $H$ given by:
\begin{align}
\pi(\lambda,m,a,b) &:=\left[\!
\begin{array}{c|ccc}
\lambda & \;0\; & \;0\; & \;0\; \\
\hline
\begin{matrix} \;0\; \\ 0 \\ 0 \end{matrix} && m
\end{array}
\!\right]
\otimes e_{22}\otimes 1 
\notag\\
& \qquad +a\otimes e_{11}\otimes e_{11} 
       +b\otimes e_{11}\otimes (1-e_{11}) \;, \label{eq:ClDA}
\end{align}
for all $\lambda\in\C$, $m\in M_3(\C)$ and $a,b\in M_4(\C)$. 
Let us identify $B$ with its representation, and omit the representation symbol.
Using \eqref{eq:circle} one easily checks $B^\circ\subseteq B'$.

In the notations of \eqref{eq:Hi}, the representation of $B$ is equivalent to the one (by matrix multiplication on the first leg)
on:
$$
(\C \otimes \C^4) \oplus (\C^3 \otimes \C^4) \oplus 
(\C^4 \otimes \C) \oplus (\C^4 \otimes \C^3).
$$
From the Lemma~\ref{lemma:11}, $B'\simeq M_4(\C) \oplus M_4(\C) \oplus\C \oplus M_3(\C) \simeq B$ and we have $B^\circ=B'$.

\begin{thm}\label{thm:15}
Let $D_0$ be as in \S\ref{sec:4} with $\delta_{ij}=0$ for all $i,j$ (Prop.~\ref{prop:5}, case 1). Then the \morita\ is satisfied if and only if:
(i) the matrices
\begin{equation}\label{eq:matrices}
\alpha:=\left[\!\begin{array}{cc}
\alpha_{13} & \alpha_{14} \\
\alpha_{23} & \alpha_{24}
\end{array}\!\right]
\qquad\qquad
\beta:=\left[\!\begin{array}{cc}
\beta_{13} & \beta_{14} \\
\beta_{23} & \beta_{24}
\end{array}\!\right]
\end{equation}
have no zero rows and
(ii) there are no $\phi,\psi\in\R$ such that
\begin{equation}\label{eq:14}
\alpha=\left[\!\begin{array}{cc}
e^{i\phi} & \zero \\
\zero & e^{i\psi}
\end{array}\!\right]\beta \;.
\end{equation}
\end{thm}

\begin{proof}
Let $B$ be the algebra \eqref{eq:Bprime} (identified with its representation on $H$).
Note that $A\subseteq B$ and $D_0\in B$, so that $\Cl_D(A)\subseteq B$ and we can use Lemma \ref{lemma:12}(b).

From \eqref{eq:setbelow} 
and \eqref{eq:circle}
every $\xi\in A'$ can then be written as a sum:
\begin{equation}\label{eq:genB}
\begin{split}
\xi &=
e_{11}\otimes e_{11}\otimes
\left[\!
\begin{array}{c|ccc}
a_1 & & v_1 & \\
\hline
\rule[-12pt]{0pt}{30pt}w_1 && b_1 &
\end{array}
\!\right]
+
e_{22}\otimes e_{11}\otimes
\left[\!
\begin{array}{c|ccc}
a_2 & & v_2 & \\
\hline
\rule[-12pt]{0pt}{30pt}w_2 && b_2 &
\end{array}
\!\right] \\
& +(e_{33}+e_{44})\otimes e_{11}\otimes
\left[\!
\begin{array}{c|ccc}
\zero & & v_3 & \\
\hline
\rule[-12pt]{0pt}{30pt}w_3 && \zero &
\end{array}
\!\right]
+e_{11}\otimes e_{12}\otimes x
+e_{11}\otimes e_{21}\otimes y
+\eta
\end{split}
\end{equation}
with $a_i\in\C$, $v_i\in(\C^3)^*$ a row vector, $w_i\in\C^3$ a column vector, $b_i\in M_3(\C)$, $x,y\in M_4(\C)$, and $\eta\in B^\circ$. With a careful inspection of \eqref{eq:ClDA} one can verify that $\xi\in B^\circ$ if and only if $\xi-\eta=0$. On the other hand $\xi\in\Cl_D(A)'$ if and only if $[D_0,\xi]=[D_0,\xi-\eta]=0$.

We now prove that conditions (i) and (ii) are satisfied if and only if the only solution to the equation $[D_0,\xi-\eta]=0$ is the trivial one, i.e.~the one given by $a_i=b_i=v_i=w_i=0\;\forall\;i$; by the discussion above, this means that $\Cl_D(A)'=B^\circ$ and, by Lemma \ref{lemma:12}(b), that the \morita\ is satisfied.

Writing down $[D_0,\xi-\eta]$ and using the linear independence of the $e_{ij}$'s, one deduces that
the commutator vanishes if and only if:
\begin{align*}
\alpha_{1j}a_1 &=0
&
\alpha_{2j}a_2 &=0
&
\beta_{1j}b_1 &=0
&
\beta_{2j}b_2 &=0
\\
\overline{\alpha}_{1j}e_{11}x &=0
&
\overline{\beta}_{1j}(1-e_{11})x &=0
&
\alpha_{1j}ye_{11} &=0
&
\beta_{1j}y(1-e_{11}) &=0
\end{align*}
and
\begin{align*}
\beta_{1j}v_1-\alpha_{1j}v_3 &=0
&
\beta_{2j}v_2-\alpha_{2j}v_3 &=0
&
\alpha_{1j}w_1-\beta_{1j}w_3 &=0
&
\alpha_{2j}w_2-\beta_{2j}w_3 &=0
\\
\overline{\alpha}_{1j} v_1-\overline{\beta}_{1j}v_3 &=0
&
\overline{\alpha}_{2j} v_2-\overline{\beta}_{2j}v_3 &=0
&
\overline{\beta}_{1j}w_1-\overline{\alpha}_{1j} w_3 &=0
&
\overline{\beta}_{2j}w_2-\overline{\alpha}_{2j} w_3 &=0
\end{align*}
for all $j=3,4$.

The first set of equations in $a_i,b_i,x,y$ admits only the zero solution if{}f $\alpha$ and $\beta$ have no zero rows, i.e.~condition (i) is satisfied (for example, $\alpha_{13}a_1=\alpha_{14}a_1=0$ admits only the solution $a_1=0$ if{}f $\alpha_{13}$ and $\alpha_{14}$ are not both zero).

Now, under the assumption that (i) holds, we can prove that the second set of equations, in $v_i,w_i$, admits only the zero solution if{}f (ii) is satisfied. Denoting by $v_i^j$ and $w_i^j$ the $j$-th component of $v_i$ and $w_i$, first we rewrite the above equations in matrix form as:
\begin{equation}\label{eq:15}
\left[\!\begin{array}{cc}
v_1^i & 0 \\
0 & v_2^i
\end{array}\!\right]\beta=\alpha
\left[\!\begin{array}{cc}
v_3^i & 0 \\
0 & v_3^i
\end{array}\!\right] \;,
\qquad
\left[\!\begin{array}{cc}
v_3^i & 0 \\
0 & v_3^i
\end{array}\!\right] \; \beta^*
=\alpha^*
\left[\!\begin{array}{cc}
v_1^i & 0 \\
0 & v_2^i
\end{array}\!\right],
\end{equation}
for all $i=1,2,3$, plus similar equations for $w_1,w_2,w_3$.
Here $\alpha,\beta$ are the matrices in \eqref{eq:matrices},
and matrix product and hermitian conjugate is understood.
Assume that there exists $\phi,\psi$ such that \eqref{eq:14}
holds (so condition (ii) is not satisfied); then \eqref{eq:15} admits non-zero solutions given for example by
$v_1=(e^{i\phi},0,0)$,
$v_2=(e^{i\psi},0,0)$ and
$v_3=(1,0,0)$.

Conversely, suppose \eqref{eq:15} has a non-zero solution. If $v_3=0$ then from \eqref{eq:15} we deduce that $\beta$ has at least one zero row (since by hypothesis $v_1$ and $v_2$ cannot both be zero), and is in contradiction with the assumption that condition (i) is satisfied. Thus $v_3^i\neq 0$ for at least one value of $i$ and, calling $c_1:=(v_3^i)^{-1}v_1^i$ and $c_2:=(v_3^i)^{-1}v_2^i$,
from the two equalities in \eqref{eq:15} we get
$$
\alpha=
\left[\!\begin{array}{cc}
c_1 & 0 \\
0 & c_2
\end{array}\!\right]\beta \;,
\qquad\quad
\beta=
\left[\!\begin{array}{cc}
\overline{c}_1 & 0 \\ 0 & \overline{c}_2
\end{array}\!\right]
\; \alpha.
$$
Combining the two, we see that $c_1 \overline{c}_1=1= c_2 \overline{c}_2$, so both $c_i$ must be unitary, thus proving that the identity \eqref{eq:14} holds for some $\phi,\psi\in\R$.
\end{proof}

\begin{rem}\label{rem:CC}
The Dirac operator of Chamseddine-Connes (see e.g.~\cite{CM08}, or also \cite[\S5.3]{DD14}) belongs to the class described in Theorem \ref{thm:15} and
is obtained by choosing the matrices \eqref{eq:matrices} as follows:
$$
\alpha^*=\left[\!\begin{array}{cc}
\Upsilon_\nu & 0 \\
0 & \Upsilon_e
\end{array}\!\right]
\qquad\qquad
\beta^*=\left[\!\begin{array}{cc}
\Upsilon_u & 0 \\
0 & \Upsilon_d
\end{array}\!\right]
$$
Theorem \ref{thm:15} tells us that the \morita\ holds if{}f
$\Upsilon_x\neq 0\;\forall\;x\in\{\nu,e,u,d\}$ and
$$
|\Upsilon_\nu|\neq|\Upsilon_u| \qquad\text{or}\qquad
|\Upsilon_e|\neq|\Upsilon_d| \;.
$$
Hence, the noncommutative description of the Standard Model leads to
an explicit restriction of the mass parameters of leptons and quarks. Although
this does not give any precise numerical predictions, its behavior with respect to the
renormalization and possible physical consequences should be further analyzed.
\end{rem}

\begin{thm}\label{thm:17}
Let $D_0$ be as in \S\ref{sec:4} with $\alpha_{13}=\alpha_{14}=0$ and $\delta_{ij}=0$ for all $i=1,2$ and $j=2,3,4$ (Prop.~\ref{prop:5}, case 2). Then the \morita\ is satisfied if and only if
$\delta_{21}\neq 0$, $(\alpha_{23},\alpha_{24})\neq (0,0)$, and the matrix $\beta$ in \eqref{eq:matrices} has no zero row.
\end{thm}

\begin{proof}
Using \eqref{eq:permmatrix} we can transform the spectral triple in one that is unitary equivalent, but notationally simpler (clearly two unitary equivalent spectral triples both satisfy the \morita, or the both don't). Conjugation by $U$ doesn't change the algebra nor $J$ (Lemma \ref{lemma:14}), but transforms $D_0$ into:
$$
UD_0U =
\left[\!\begin{array}{cc|cc}
\zero & \overline\delta_{21} & \zero & \zero \\
\delta_{21} & \zero & \alpha_{23} & \alpha_{24} \\
\hline
\zero & \overline\alpha_{23} & \zero & \zero \\
\zero & \overline\alpha_{24} & \zero & \zero
\end{array}\!\right]
\otimes e_{11}\otimes e_{11}+
\left[\!\begin{array}{cc|cc}
\zero & \zero & \beta_{13} & \beta_{14} \\
\zero & \zero & \beta_{23} & \beta_{24} \\
\hline
\overline\beta_{13} & \overline\beta_{23} & \zero & \zero \\
\overline\beta_{14} & \overline\beta_{24} & \zero & \zero
\end{array}\!\right]
\otimes e_{11}\otimes (1-e_{11})
$$
We can repeat almost verbatim the proof of Theorem \ref{thm:15}.
Both $A$ and $UD_0U$ are contained in the algebra $B$ in \eqref{eq:Bprime}. Let $\xi$ be as in \eqref{eq:genB}.
We must prove that the equation $[UD_0U,\xi-\eta]=0$ admits only the trivial solution
$a_i=b_i=v_i=w_i=0\;\forall\;i$ if and only if the conditions in Theorem \ref{thm:17} are satisfied.
Writing down the commutator one sees that it vanishes if{}f
\begin{align*}
\delta_{21}a_1 &=0
&
\alpha_{2j}a_2 &=0
&
\beta_{1j}b_1 &=0
&
\beta_{2j}b_2 &=0
\\
\delta_{21}e_{11}x &=0
&
\overline\beta_{1j}(1-e_{11})x &=0
&
\overline\delta_{21}ye_{11} &=0
&
\beta_{1j}y(1-e_{11}) &=0
\end{align*}
and
\begin{align*}
\delta_{21}v_1 &=0
&
\beta_{1j}v_1 &=0
&
\overline\delta_{21}w_1 &=0
&
\overline\beta_{1j}w_1 &=0
\\
\overline\delta_{21}v_2 &=0
&
\beta_{2j}v_2 &=\alpha_{2j}v_3
&
\delta_{21}w_2 &=0
&
\alpha_{2j}w_2 &=\beta_{2j}w_3
\\
\overline\beta_{1j}v_3 &=0
&
\overline\alpha_{2j}v_2 &=\overline\beta_{2j}v_3
&
\beta_{1j}w_3 &=0
&
\overline\beta_{2j}w_2 &=\overline\alpha_{2j}w_3
\end{align*}
for all $j=3,4$. The first set of equations has only the zero solution if{}f $\delta_{21}\neq 0$, $(\alpha_{23},\alpha_{24})\neq (0,0)$ and the first row of $\beta$ is not zero. In addition, the second set of equations admits only the zero solution if{}f the second row of $\beta$ is also not zero.
\end{proof}

\subsection{Cases 3 and 4.}

For $\lambda\in\C$, $m\in M_3(\C)$ and $i=1,2$ define
\begin{align*}
\pi_0(\lambda) &:=\lambda e_{11}\otimes (1\otimes 1-e_{11}\otimes e_{11}) \;,\\
\pi_i(m) &:=
\left[\!
\begin{array}{c|ccc}
0 & \;0\; & \;0\; & \;0\; \\
\hline
\begin{matrix} \;0\; \\ 0 \\ 0 \end{matrix} && m
\end{array}
\!\right]
\otimes e_{ii}\otimes (1-e_{11}) \;,
\end{align*}
and let $\pi_3$ be the representation of $M_7(\C)$ given by
\begin{align*}
\pi_3(a) &:=
\left[\!\begin{array}{cccc}
a_{11} & a_{12} & a_{13} & a_{14} \\
a_{21} & a_{22} & a_{23} & a_{24} \\
a_{31} & a_{32} & a_{33} & a_{34} \\
a_{41} & a_{42} & a_{43} & a_{44}
\end{array}\!\right]
\otimes e_{11}\otimes e_{11}+
\left[\!\begin{array}{c|ccc}
\zero & a_{15} & a_{16} & a_{17} \\
\zero & a_{25} & a_{26} & a_{27} \\
\zero & a_{35} & a_{36} & a_{37} \\
\zero & a_{45} & a_{46} & a_{47}
\end{array}\!\right]
\otimes e_{12}\otimes e_{11} \\
&\; +
\left[\!\begin{array}{cccc}
\zero & \zero & \zero & \zero \\
\hline
a_{51} & a_{52} & a_{53} & a_{54} \\
a_{61} & a_{62} & a_{63} & a_{64} \\
a_{71} & a_{72} & a_{73} & a_{74}
\end{array}\!\right]
\otimes e_{21}\otimes e_{11}+
\left[\!\begin{array}{c|ccc}
\zero & \zero & \zero & \zero \\
\hline
\zero & a_{55} & a_{56} & a_{57} \\
\zero & a_{65} & a_{66} & a_{67} \\
\zero & a_{75} & a_{76} & a_{77}
\end{array}\!\right]
\otimes e_{22}\otimes e_{11}
\end{align*}
for all $a=(a_{ij})\in M_7(\C)$.
The product of any two of these representations is zero, and their sum gives a faithful unital representation of the algebra
\begin{equation}\label{eq:Bsecond}
B:=\C\oplus M_3(\C)\oplus M_3(\C)\oplus M_7(\C) \;.
\end{equation}
Using \eqref{eq:circle}, one easily checks that $\pi_i$ commutes with $\pi_j^\circ$ for all $i,j$. Therefore, after identifying $B$ with its representation, we can conclude that $B^\circ\subseteq B'$. 

In the notations of \eqref{eq:Hi}, the representation of $B$ is equivalent to the one (by matrix multiplication on the first leg) on
$$
(\C\otimes\C^7)\oplus
(\C^3\oplus\C^3)\oplus
(\C^3\oplus\C^3)\oplus
(\C^7\oplus\C) \;.
$$
From the Lemma~\ref{lemma:11}, $B'\simeq M_7(\C)\oplus M_3(\C)\oplus M_3(\C)\oplus\C\simeq B$, and so $B^\circ=B'$.

\begin{thm}\label{thm:18}
Let $D_0$ be as in \S\ref{sec:4} with $\delta_{21} = 0$ and $\beta_{13}=\beta_{14}=0$ (Prop.~\ref{prop:5}, case 3). Then the \morita\ is satisfied if and only if
$(\beta_{23},\beta_{24})\neq (0,0)$ and of the four vectors
\begin{equation}\label{eq:vectors}
(\alpha_{13},\alpha_{14})
\;,\qquad
(\alpha_{23},\alpha_{24})
\;,\qquad
(\delta_{12},\delta_{13},\delta_{14})
\;,\qquad
(\delta_{22},\delta_{23},\delta_{24}) \;,
\end{equation}
at least three are not zero.
\end{thm}
\begin{proof}
The strategy of the proof is the same as for Theorems \ref{thm:15} and \ref{thm:17}.
Let $B$ be the algebra \eqref{eq:Bsecond} (identified with its representation on $H$).
Note that $A\subseteq B$ and $D_0\in B$, so that $\Cl_D(A)\subseteq B$. From \eqref{eq:setbelow}
and \eqref{eq:circle}
every $\xi\in A'$ can then be written as a sum:
\begin{align}
\xi &=
e_{11}\otimes e_{11}\otimes
\left[\!
\begin{array}{c|ccc}
\zero & & v_1 & \\
\hline
\rule[-12pt]{0pt}{30pt}w_1 && \zero &
\end{array}
\!\right]
+e_{11}\otimes e_{12}\otimes \left[\!
\begin{array}{c|ccc}
a_2 & & v_2 & \\
\hline
\rule[-12pt]{0pt}{30pt}\zero && \zero &
\end{array}
\!\right]
\notag \\
&+e_{11}\otimes e_{21}\otimes
\left[\!
\begin{array}{c|ccc}
a_3 & & \zero & \\
\hline
\rule[-12pt]{0pt}{30pt}w_3 && \zero &
\end{array}
\!\right]
+e_{22}\otimes e_{11}\otimes
\left[\!
\begin{array}{c|ccc}
a_4 & & v_4 & \\
\hline
\rule[-12pt]{0pt}{30pt}w_4 && b_4 &
\end{array}
\!\right] \label{eq:21} \\
&+(e_{33}+e_{44})\otimes e_{11}\otimes
\left[\!
\begin{array}{c|ccc}
a_5 & & v_5 & \\
\hline
\rule[-12pt]{0pt}{30pt}w_5 && \zero &
\end{array}
\!\right]
+(1-e_{11})\otimes e_{22}\otimes\left[\!
\begin{array}{c|ccc}
a_6 & & v_6 & \\
\hline
\rule[-12pt]{0pt}{30pt}w_6 && \zero &
\end{array}
\!\right]
+\eta \notag
\end{align}
with $a_i\in\C$, $v_i\in(\C^3)^*$ a row vector, $w_i\in\C^3$ a column vector, $b_i\in M_3(\C)$ and $\eta\in B^\circ$.
The element $\xi$ belongs to $B^\circ$ if{}f $\xi-\eta=0$.

As in previous cases, one transforms $[D_0,\xi-\eta]=0$ into a system of linear equations and checks that this admits only the zero solution
$a_i=b_i=v_i=w_i=0\;\forall\;i$ if and only if the conditions in Theorem \ref{thm:18} are satisfied (this is a tedious computation similar to those for \ref{thm:15} and \ref{thm:17}, that we omit).
\end{proof}

\begin{thm}\label{thm:19}
Let $D_0$ as in \S\ref{sec:4} with $\delta_{12}=\delta_{13}=\delta_{14}=\beta_{13}=\beta_{14}=\alpha_{13}=\alpha_{14}=0$ (Prop.~\ref{prop:5}, case 4). Then the \morita\ is satisfied if and only if
$\delta_{21}\neq 0$ and none of the following three vectors
$$
(\alpha_{23} , \alpha_{24})
\;,\qquad
(\beta_{23}, \beta_{24})
\;,\qquad
(\delta_{22} , \delta_{23} , \delta_{24}) \;,
$$
is zero.
\end{thm}
\begin{proof}
Let $B$ be the algebra \eqref{eq:Bsecond}. Using \eqref{eq:permmatrix} we transform the spectral triple in one that is unitary equivalent. Conjugation by $U$ doesn't change the algebra nor $J$ (Lemma \ref{lemma:14}), but transforms $D_0$ into:
\begin{align*}
\hspace*{-5mm}UD_0U &=
\left[\!\begin{array}{cc|cc}
\zero & \overline\delta_{21} & \zero & \zero \\
\delta_{21} & \zero & \alpha_{23} & \alpha_{24} \\
\hline
\zero & \overline\alpha_{23} & \zero & \zero \\
\zero & \overline\alpha_{24} & \zero & \zero
\end{array}\!\right]
\otimes e_{11}\otimes e_{11}+
\left[\!\begin{array}{cc|cc}
\zero & \zero & \zero & \zero \\
\zero & \zero & \beta_{23} & \beta_{24} \\
\hline
\zero & \overline\beta_{23} & \zero & \zero \\
\zero & \overline\beta_{24} & \zero & \zero
\end{array}\!\right]
\otimes e_{11}\otimes (1-e_{11}) \\
&\; +
\left[\!\begin{array}{cc|cc}
\zero & \zero & \zero & \zero \\
\zero & \delta_{22} & \delta_{23} & \delta_{24} \\
\hline
\zero & \zero & \zero & \zero \\
\zero & \zero & \zero & \zero
\end{array}\!\right]
\otimes e_{12}\otimes e_{11}+
\left[\!\begin{array}{cc|cc}
\zero & \zero & \zero & \zero \\
\zero & \overline\delta_{22} & \zero & \zero \\
\hline
\zero & \overline\delta_{23} & \zero & \zero \\
\zero & \overline\delta_{24} & \zero & \zero
\end{array}\!\right]
\otimes e_{21}\otimes e_{11}
\end{align*}
Both $A$ and $UD_0U$ are contained in the algebra $B$, and we can repeat once again the proof of previous three theorems.

For $\xi$ be as in \eqref{eq:21}, one checks that the equation $[D_0,\xi-\eta]=0$ admits only the zero solution
$a_i=b_i=v_i=w_i=0\;\forall\;i$ if and only if the conditions in Theorem \ref{thm:19} are satisfied.
\end{proof}

\begin{rem}
In each of the four cases in Prop.~\ref{prop:5}: Dirac operators satisfying the \secondorder\ are parametrized by a finite-dimensional complex vector space; Dirac operators not satisfying the \morita\ form a submanifold of codimension $\geq 1$.
In this sense, we can say that 
the ``generic'' Dirac operator satisfying the \secondorder\ satisfies the \morita\ as well (those not satisfying it being ``exceptions'').
\end{rem}

\section{On generations}\label{sec:6}
Let us conclude with a few comments on what happens with more than one generation of particles.
In Prop.~\ref{prop:5} we reproduced the results on the \secondorder\ in the penultimate section of \cite{FB14a}, under the assumption of a single generation of fermions. Though in \cite{FB14a} it is remarked that the extension to the three generations is straightforward, this requires a further scrutiny.

In the case of $n$ generations (e.g.~$n=3$), the Hilbert space in \S\ref{sec:4} is tensored by $\C^n$, $J$ acts on this additional factor by component-wise complex conjugation, and the entries $\alpha_{ij},\beta_{ij},\delta_{ij}$ of $D_0$ become $n\times n$ complex matrices acting on $\C^n$. One can repeat the proof of Prop.~\ref{prop:5} and find that the \secondorder\ is satisfied if{}f the matrices
\begin{equation}\label{eq:prod}
\delta_{21}\overline\alpha_{1i} \;,\qquad
\delta_{21}\overline\delta_{1j} \;,\qquad
\delta_{ij}\overline\beta_{13} \;,
\end{equation}
vanish for all $i=1,2$ and $j=2,3,4$.

The difference is that then (if $n>1$) we cannot use the cancellation property to conclude that in the products \eqref{eq:prod} at least one of the two factors must be zero. Thus, we don't get the four cases in Prop.~\ref{prop:5} anymore: those conditions become sufficient but no longer necessary.
In the special case of Chamseddine-Connes Dirac operator, described in Remark \ref{rem:CC} (with $\Upsilon_x$ matrices acting on the additional $\C^n$ factor), the \secondorder\ is however satisfied. Whether the \morita\ is satisfied (for almost all $\Upsilon_x$'s) is under investigation.

\section*{Acknowledgement}
L.D.\ is grateful for his support at IMPAN provided by Simons-Foundation grant 346300 and a Polish Government MNiSW 2015-2019 matching fund. A.S.\ acknowledges the support from the grant NCN 2015/19/B/ST1/03098. This work is part of the project Quantum Dynamics sponsored by the EU-grant RISE 691246 and Polish Government grant 317281.

\end{document}